\documentclass[sigconf,nonacm,screen]{acmart}

\usepackage{colortbl}
\usepackage{pifont}
\usepackage{makecell}
\usepackage{xspace}
\usepackage{balance}
\usepackage{booktabs}
\usepackage{multirow}

\definecolor{myred}{rgb}{0.7, 0.3, 0.0}
\definecolor{myblue}{HTML}{0a41b8}
\definecolor{mygreen}{HTML}{056b34}
\definecolor{mypurple}{HTML}{5d1e8b}
\definecolor{rowgray}{HTML}{f0f0f0}
\definecolor{rowhighlight}{HTML}{f0edff}
\newcommand{\best}[1]{\textbf{#1}}

\newcommand{\method}{HyPASE\xspace}

\AtBeginDocument{%
  }

\copyrightyear{2026}
\acmYear{2026}
\setcopyright{none}
\acmConference[MM '26]{ACM Multimedia 2026}{November 10--14, 2026}{Rio de Janeiro, Brazil}
\acmBooktitle{ACM Multimedia 2026}
\renewcommand\footnotetextcopyrightpermission[1]{}
\begin{document}

\title{HyPASE: Hyperbolic Geometry for Parameter-Efficient Speech Emotion Fine-Tuning Framework for Large Audio-Language Models}


\author{Tian Jin}
\orcid{0009-0008-1037-6729}
\affiliation{%
  \institution{Tongji University}
  \city{Shanghai}
  \country{China}
}
\affiliation{%
  \institution{The Chinese University of Hong Kong, Shenzhen}
  \city{Shenzhen}
  \country{China}
}
\email{226085012@link.cuhk.edu.cn}

\author{Ruikang Zhang}
\orcid{0009-0007-5284-6090}
\affiliation{%
  \institution{Tongji University}
  \city{Shanghai}
  \country{China}
}
\affiliation{%
  \institution{Peking University}
  \city{Beijing}
  \country{China}
}
\email{2601213508@stu.pku.edu.cn}

\author{Zefeng Zhao}
\authornote{Corresponding authors: Zefeng Zhao and Jin Zeng.}
\orcid{0000-0002-5859-2622}
\affiliation{%
  \institution{The Chinese University of Hong Kong, Shenzhen}
  \city{Shenzhen}
  \country{China}
}
\email{zefengzhao@cuhk.edu.cn}

\author{Ding Luo}
\orcid{0009-0001-3611-8284}
\affiliation{%
  \institution{Tongji University}
  \city{Shanghai}
  \country{China}
}
\email{2452687@tongji.edu.cn}

\author{Jin Zeng}
\authornotemark[1]
\orcid{0000-0002-0180-7733}
\affiliation{%
  \institution{Tongji University}
  \city{Shanghai}
  \country{China}
}
\email{zengjin@tongji.edu.cn}

\renewcommand{\shortauthors}{Tian et al.}



\begin{abstract}
Large Audio-Language Models (LALMs) excel at general speech understanding. However, adapting them to fine-grained tasks like Speech Emotion Recognition (SER) remains a significant bottleneck. Current Parameter-Efficient Fine-Tuning (PEFT) methods typically operate in flat Euclidean space. This geometry fails to capture the multi-granularity nature of emotion cues, which range from low-level prosody to high-level semantics. To address this, we propose \method, a hyperbolic PEFT framework for LALM-based SER. \method leverages the Poincar\'{e} ball model, using the hyperbolic radius as an explicit proxy for representational granularity. The framework consists of two core components: a Hyperbolic Geometric Adapter (HGA) for layer-adaptive weight modulation, and an Emotion-aware Multi-capacity Cross-modal Aggregator (EMCA) that compresses multi-scale features into compact audio prefixes. Empirical results on standard benchmarks show that \method outperforms Euclidean PEFT baselines across all metrics on MELD and achieves a notable Unweighted Accuracy gain on IEMOCAP, particularly in class-imbalanced emotion recognition, with the accompanying slight Weighted Accuracy trade-off reflecting hyperbolic space's geometric prioritization of minority-class representations. Furthermore, \method achieves robust zero-shot cross-dataset generalization within a constrained parameter budget. By grounding the adaptation process in hyperbolic geometry, \method offers a highly efficient path for LALM fine-tuning.
\end{abstract}



\maketitle

\section{Introduction}

Large Audio-Language Models (LALMs) deeply couple pre-trained acoustic encoders with large-scale language models, and demonstrate strong zero-shot understanding capabilities in tasks such as speech instruction following, audio description, and cross-modal reasoning~\cite{qwen2audio2024,salmonn2023}. Systems represented by Qwen2-Audio and SALMONN have achieved unified modeling of acoustic perception and language reasoning, establishing LALMs as a foundational framework for speech understanding. Nevertheless, how to transfer this general perception capability to fine-grained emotion understanding tasks such as Speech Emotion Recognition (SER) remains an underexplored problem. Recent LALM-based SER studies have explored speech-aware Q-Former alignment, context-aware instruction tuning, and chain-of-thought-based stabilization~\cite{qformer_ser2024,du2025eaa,zhao2025c2ser,wang2025cot4ad}, but how to exploit the inherent geometric structure of emotion representations during parameter-efficient adaptation remains open.

The central challenge lies in the inherent hierarchical nature of emotion cues: spanning from global prosodic contours to local spectral variations and higher-level semantic context, these signals are structured across multiple granularity levels~\cite{scherer1986,banse1996,gobl2003,egemaps2016,larrouy2024,tzirakis2021,dutta2024}. In LALM encoders, this hierarchy is mirrored by the progressive evolution of representations from shallow acoustic templates to deep emotion-discriminative abstractions. Despite this, existing PEFT methods collapse these rich structures into flat Euclidean embeddings by applying a uniform adaptation strategy without any awareness of representational depth, leaving a consistent performance gap versus full fine-tuning on standard benchmarks such as IEMOCAP~\cite{iemocap2008} and MELD~\cite{meld2019,mstr2024,emosllm2025}. How to capture this multi-granularity hierarchy in a geometrically explicit manner within a limited parameter budget constitutes the central challenge for LALM emotion adaptation.

While Parameter-Efficient Fine-Tuning (PEFT) has emerged as the standard for model adaptation, and methods such as LoRA~\cite{lora2022}, Adapter~\cite{adapters2019}, and Prefix Tuning~\cite{prefix2021} have been applied to LALM emotion adaptation~\cite{emosllm2025,dpm2025}, their reliance on Euclidean transformations fundamentally lacks the geometric inductive bias needed to preserve the hierarchical depth of audio-language representations~\cite{poincare2017}. Applying the same adaptation strategy uniformly across all encoder layers leaves two geometric gaps unaddressed. First, Euclidean adaptation modifies representations without any reference to their hierarchical depth, so each layer's adaptation magnitude remains unconstrained relative to the encoder hierarchy. Second, Euclidean pooling of frame representations discards radial distinctions that encode granularity levels, collapsing multi-granularity structure before it can inform emotion recognition.

The two gaps above point to a common need: a geometric framework capable of natively encoding the hierarchical depth of encoder representations. A direct empirical answer comes from Gromov $\delta$-hyperbolicity analysis on Qwen2-Audio encoder representations: all layers exhibit substantially lower $\delta$ than random Gaussian embeddings, confirming that the encoder's learned geometry is intrinsically tree-like rather than flat. Hyperbolic space, a Riemannian manifold of constant negative curvature, accommodates tree-like hierarchical structures through exponential volume growth with radius~\cite{poincare2017,hgcn2019}. In the Poincar\'{e} ball, geodesic distance from the origin encodes hierarchical depth: general patterns cluster at smaller radii, while fine-grained discriminative semantics lie at larger radii, directly mirroring the progressive abstraction in LALM encoders. Extending PEFT to hyperbolic space therefore enables geometric alignment of the emotion representation hierarchy with minimal parameters. This idea has received initial validation in vision-language models~\cite{hyperet2024} and multimodal safety alignment~\cite{hysac_cvpr2025}, but remains unexplored for LALM adaptation.

This gap is the starting point of the present work. We propose \method, which, to the best of our knowledge, is the first framework to integrate hyperbolic geometry directly into LALM parameter adaptation. The core geometric contribution is the Hyperbolic Geometric Adapter (HGA), which performs layer-adaptive hyperbolic radius modulation in the weight space of the audio encoder. Complementing HGA, the Emotion-aware Multi-capacity Cross-modal Aggregator (EMCA) is a lightweight utterance-level readout module. It compresses HGA-enhanced frame representations into compact audio prefixes for the frozen language model through task-conditioned multi-branch aggregation and hyperbolic fusion. The main contributions of this paper are:

(1) \textbf{A Geometry-Grounded Framework.} We propose \method, introducing hyperbolic geometry into PEFT for LALMs to enable parameter-efficient Speech Emotion Recognition.

(2) \textbf{A Two-Level Adaptation Pipeline.} We design an end-to-end pipeline featuring a Hyperbolic Geometric Adapter (HGA) for layer-wise weight modulation, and an EMCA module to fuse multi-scale features into structured audio prefixes.

(3) \textbf{Minority-Class-Aware Efficiency and Generalization.} Empirical results show that \method outperforms standard Euclidean PEFT methods on MELD across all metrics and substantially improves minority-class recognition on IEMOCAP (Unweighted Accuracy), with the accompanying Weighted Accuracy trade-off being an intentional consequence of hyperbolic space's prioritization of underrepresented classes; \method also delivers robust zero-shot generalization, all within a minimal parameter footprint.

\begin{figure*}[t!]
  \centering
  \includegraphics[width=\textwidth]{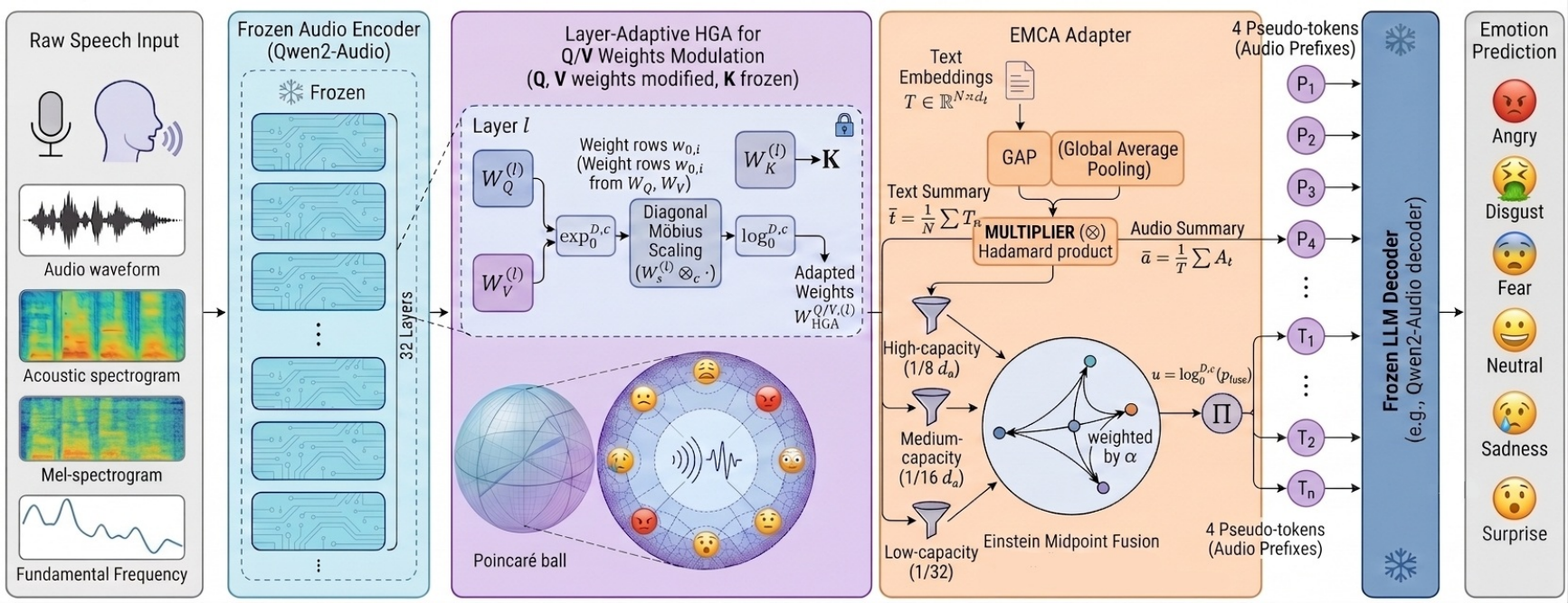}
  \caption{Overall framework of \method. HGA is injected into each layer of the frozen audio encoder, mapping Q/V weight rows through $\exp_{\mathbf{0}}^{D,c}$, M\"{o}bius scaling, and $\log_{\mathbf{0}}^{D,c}$ for layer-adaptive hyperbolic radius modulation. EMCA converts audio frame representations and instruction prompt embeddings into $K$ audio prefix tokens: it applies task-conditioned gating, compresses through three capacity branches (high/medium/low), fuses via the Einstein midpoint, and expands into prefix tokens. The prefix tokens are projected by the multimodal projector and fed into the frozen language model alongside text token embeddings to generate emotion labels.}
  \label{fig:overview}
\end{figure*}

\section{Related Work}

\subsection{Speech Emotion Recognition}

Speech emotion recognition has undergone a paradigm shift from hand-crafted acoustic features to self-supervised representation learning. Low-level acoustic feature methods exemplified by eGeMAPS~\cite{egemaps2016} established early baselines. Self-supervised pre-training models including HuBERT~\cite{hsu2021}, WavLM~\cite{chen2022wavlm}, and Emotion2vec~\cite{ma2023emotion2vec} advanced SER performance by learning general representations on large unlabeled speech corpora. Multi-scale temporal modeling~\cite{mstr2024} and multi-granularity cross-modal fusion~\cite{mseadapter2024} further confirmed the critical role of hierarchical emotion cues in recognition performance. More recently, LALM-based SER has attracted growing attention along three technical directions. Alignment-oriented methods such as EmoQ~\cite{qformer_ser2024} and EAA~\cite{du2025eaa} use Q-Former or cross-attention bridges for emotion-conditioned feature extraction; instruction-tuning approaches including EmoSLLM~\cite{emosllm2025} and C$^2$SER~\cite{zhao2025c2ser} exploit chain-of-thought prompting for more stable predictions; and reinforcement-based methods such as DPM~\cite{dpm2025}, EMO-RL~\cite{li2025emorl}, and EmoAgent-R1~\cite{fang2026emoagent} refine emotion outputs through reward shaping and dynamic agent specialization. Regardless of their specific strategies, however, all these models treat the audio encoder as a static Euclidean feature extractor, ignoring the geometric potential for representation reshaping during adaptation.

\subsection{Large Audio-Language Models}

Research coupling large-scale language models with pre-trained acoustic encoders has advanced rapidly. Recent surveys frame LALMs as reasoning-centric multimodal systems spanning auditory processing, dialogue, and trustworthiness~\cite{yang2025lalmsurvey,xiao2026dualstreamdecoupledlearningtemporal}. Whisper~\cite{radford2023robust} established a strong acoustic feature extraction foundation through weak supervision on large-scale multilingual data. SALMONN~\cite{salmonn2023} bridges dual acoustic encoders with an LLM through a Q-Former for unified audio-language reasoning. Qwen2-Audio~\cite{qwen2audio2024} proposes a unified multi-task audio understanding framework covering speech recognition, audio description, and emotion understanding in a single model. These advances notwithstanding, prevailing PEFT methods such as LoRA and Adapter operate as Euclidean affine transformations and cannot reshape the geometric hierarchy of encoder representations, leaving the intrinsic hyperbolic structure of multi-granularity emotion cues unexploited during adaptation.

\subsection{Hyperbolic Learning Methods}

Hyperbolic geometry has attracted growing attention in deep learning, driven by its ability to represent hierarchical structures with low distortion~\cite{poincare2017}. HGCN~\cite{hgcn2019} introduced hyperbolic operations into graph convolutional networks and achieved notable improvements on knowledge graphs and hierarchical classification tasks. In multimodal foundation models, HyperET~\cite{hyperet2024} introduced hyperbolic geometric adapters for image-text alignment, and HySAC~\cite{hysac_cvpr2025} applied hyperbolic space to multimodal safety alignment. In multimodal SER, HyFuse~\cite{hyfuse2025,fu2026emotioncolliderdualhyperbolic} incorporated hyperbolic fusion into the feature aggregation stage to better integrate heterogeneous emotion cues from multiple modalities. Unlike these contemporary works, which apply hyperbolic geometry only at the downstream fusion stage, \method performs a stronger intervention by directly modulating the internal weights of the frozen backbone through HGA, reshaping representations at the encoder level rather than post-hoc~\cite{cao2025causalctrl,ding2025dualsg}.

\section{Method}

\method has two trainable components: HGA, which modulates the hyperbolic radius of weight row embeddings in each audio encoder layer, and EMCA, a lightweight module that converts frame representations into compact audio prefixes via task-conditioned gating and hyperbolic fusion. The audio encoder and language decoder are frozen; only HGA (Q/V projections), residual adapters, EMCA, and the native multimodal projector are trained. Euclidean counterparts and disabled losses are explored as ablation conditions in Section~\ref{sec:ablation}.

\subsection{Hyperbolic Geometry Preliminaries}
\label{sec:prelim}

\subsubsection{Poincar\'{e} Ball and Hyperbolic Radius}

We work in the $d$-dimensional Poincar\'{e} ball $\mathcal{D}_c^d := \{\mathbf{x} \in \mathbb{R}^d : c\lVert \mathbf{x} \rVert^2 < 1\}$ with curvature $-c$ ($c > 0$). The key quantity for our framework is the \emph{hyperbolic radius}---the geodesic distance from a point to the origin:
\begin{equation}
  \operatorname{Rad}_{\mathbf{x}} := d_c(\mathbf{x},\, \mathbf{0}) = \frac{2}{\sqrt{c}} \operatorname{artanh}(\sqrt{c}\lVert \mathbf{x} \rVert).
  \label{eq:hyp_radius}
\end{equation}

$\operatorname{Rad}_{\mathbf{x}}$ naturally characterizes the granularity hierarchy of representations: a smaller radius corresponds to more coarse-grained shared patterns, while a larger radius corresponds to more fine-grained discriminative semantics, closely aligned with the feature abstraction progression from shallow to deep in audio encoders. Both HGA radius modulation and EMCA branch radius ordering are grounded in this geometric interpretation.

\subsubsection{Hyperbolic Operations}

Features are mapped between the Euclidean tangent space at the origin and the Poincar\'{e} ball via the standard exponential map $\exp_{\mathbf{0}}^{D,c}$ and logarithmic map $\log_{\mathbf{0}}^{D,c}$; M\"{o}bius scalar multiplication $s \otimes_c \mathbf{x}$ and diagonal-matrix multiplication $\mathbf{D} \otimes_c \mathbf{x}$ are defined through these maps. All hyperbolic operations in HGA act independently on each row of the weight matrix.

\subsection{Overall Architecture}
\label{sec:overall}

The frozen audio encoder ($L$ layers, hidden dimension $d_a$, each injected with HGA) encodes audio input $\mathbf{x}_a$ into frame-level states $\mathbf{A} \in \mathbb{R}^{T \times d_a}$. The frozen LM embedding layer maps the emotion classification instruction prompt $\mathbf{x}_t$ (a fixed task-specific template specifying the target emotion categories; not ASR transcription of the speech) to token embeddings $\mathbf{T} \in \mathbb{R}^{N \times d_t}$. EMCA takes $(\mathbf{A}, \mathbf{T})$ and produces $K$ audio prefix tokens via task-conditioned gating, three-branch utterance-level aggregation, and hyperbolic fusion; the multimodal projector $\operatorname{MMP}(\cdot;\theta_{\text{proj}})$ maps them to $\mathbf{H}_{\text{audio}} \in \mathbb{R}^{K \times d_t}$ for input to the frozen LM:
\begin{equation}
  \hat{\mathbf{y}} = \operatorname{LLM}\bigl([\mathbf{H}_{\text{audio}};\, \mathbf{T}];\, \theta_{\text{LLM}}\bigr).
  \label{eq:llm}
\end{equation}

\begin{figure}[t]
  \centering
  \includegraphics[width=\columnwidth]{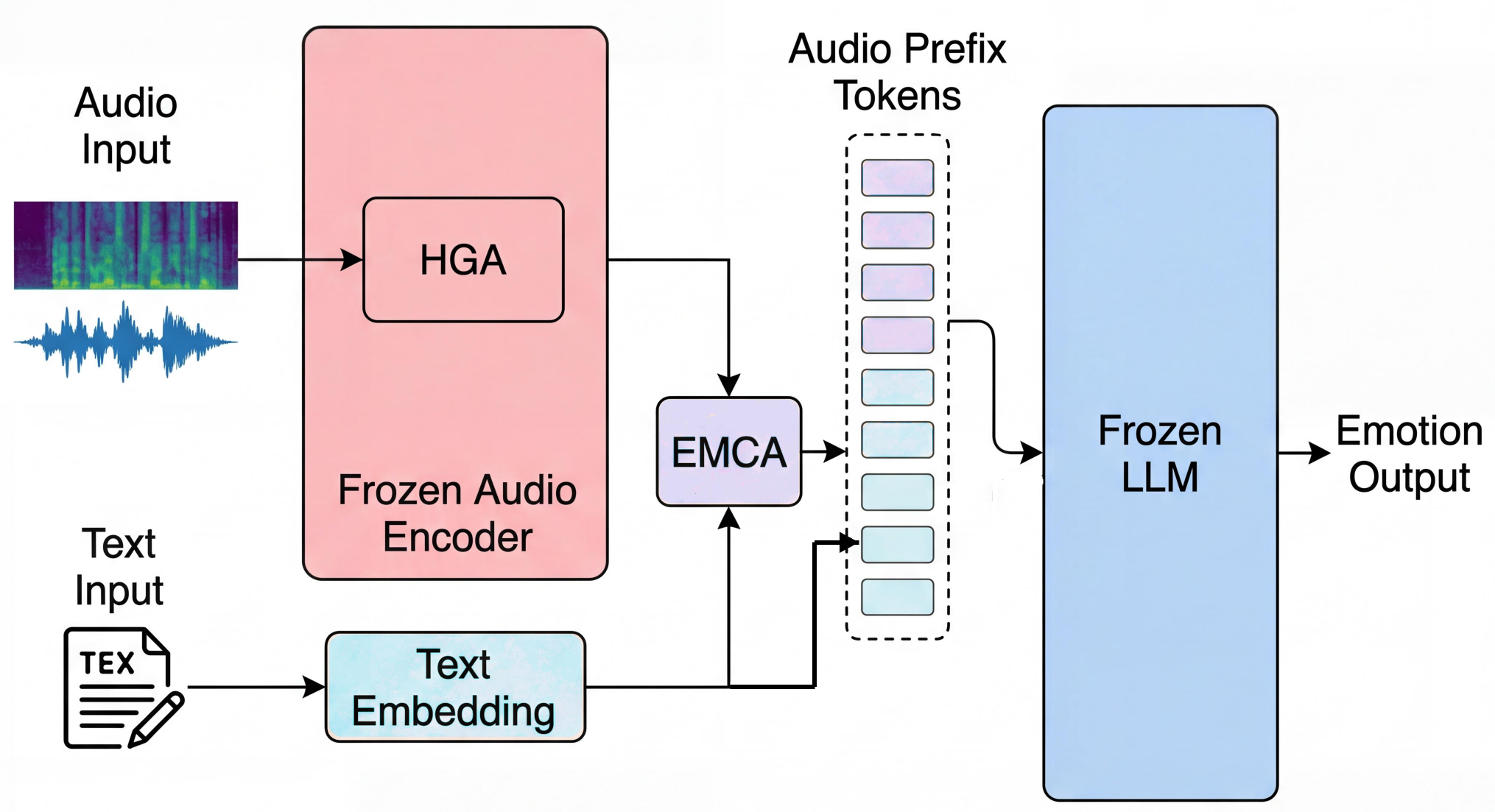}
  \caption{Method details of \method. Left: HGA maps frozen Q/V weight rows into the Poincar\'{e} ball via $\exp_{\mathbf{0}}^{D,c}$, applies diagonal M\"{o}bius scaling $\mathbf{W}_s^{(\ell)}\otimes_c(\cdot)$ parameterized by $\mathbf{s}^{(\ell)}$, and recovers adapted weights via $\log_{\mathbf{0}}^{D,c}$, achieving layer-adaptive hyperbolic radius modulation. Right: EMCA compresses the task-conditioned audio summary through three capacity branches ($d_a/8$, $d_a/16$, $d_a/32$), fuses them via the Einstein midpoint operator, and the prefix generator produces $K$ audio prefix tokens.}
  \label{fig:flowchart}
\end{figure}

\subsection{Hyperbolic Geometric Adapter (HGA)}
\label{sec:hga}

HGA operates in the weight space of the audio encoder, achieving geometrically persistent radius modulation without per-sample inference overhead. For the pre-trained weight $\mathbf{W}_0^{(\ell)} \in \mathbb{R}^{d_a \times d_a}$ of layer $\ell$, let $\mathbf{w}_{0,i}^{(\ell)} \in \mathbb{R}^{d_a}$ denote its $i$-th row. Its row-wise hyperbolic embedding is $\mathbf{w}_{D,0,i}^{(\ell)} = \exp_{\mathbf{0}}^{D,c}(\mathbf{w}_{0,i}^{(\ell)}) \in \mathcal{D}_c^{d_a}$. HGA modulates these row embeddings by adjusting their hyperbolic radii.

\subsubsection{Hyperbolic Radius Modulation: From Scalar Radius Scaling to Diagonal Radius Modulation}

The starting point of HGA is a row-wise scalar radius scaling constraint on the hyperbolic embeddings of the weight rows:
\begin{equation}
  \operatorname{Rad}_{\mathbf{w}_{D,i}^{(\ell)}} = s^{(\ell)} \cdot \operatorname{Rad}_{\mathbf{w}_{D,0,i}^{(\ell)}}, \quad \forall\, i,
  \label{eq:radius_constraint}
\end{equation}
where the task-adaptive scaling coefficient $s^{(\ell)} > 0$ determines the direction and magnitude of radial displacement for layer $\ell$ weights in hyperbolic space.

\begin{theorem}[Radius scaling equivalence of M\"{o}bius scalar multiplication]
\label{thm:radius_scaling}
For any point $\mathbf{X}$ in $\mathcal{D}_c^d$ and scalar $s > 0$, scalar M\"{o}bius multiplication satisfies:
\begin{equation}
  s \cdot \operatorname{Rad}_{\mathbf{X}} = \operatorname{Rad}_{s \otimes_c \mathbf{X}}.
  \label{eq:mobius_scalar}
\end{equation}
\end{theorem}
By Theorem~\ref{thm:radius_scaling}, the scalar constraint in Equation~(\ref{eq:radius_constraint}) is realized by row-wise scalar M\"{o}bius multiplication, allowing the Euclidean scalar $s^{(\ell)}$ to be optimized via standard AdamW while strictly realizing the intended radial displacement. However, a single scalar $s^{(\ell)}$ applies uniform radial scaling to all rows and all feature dimensions, which is too rigid for adaptation across heterogeneous encoder channels. Different layers of a frozen audio encoder capture emotion cues at different hierarchical levels: shallow layers encode acoustic templates such as pitch contour and energy envelope, while deep layers distill emotion-discriminative semantics. Forcing a uniform adaptation across this hierarchy is therefore sub-optimal. Our solution is to generalize the scalar rule to a diagonal modulation matrix $\mathbf{W}_s^{(\ell)} = \operatorname{diag}(\mathbf{s}^{(\ell)})$ with $\mathbf{s}^{(\ell)} \in \mathbb{R}_{>0}^{d_a}$, which acts as a layer-adaptive granularity controller: by learning per-layer, per-dimension scaling factors, HGA adjusts the radial position of representations layer by layer, explicitly aligning each layer's adaptation with the encoder's intrinsic hierarchy. The resulting Euclidean reparameterization of HGA is:
\begin{equation}
  \mathbf{W}_{\text{HGA}}^{(\ell)} := \left[\log_{\mathbf{0}}^{D,c}\!\Bigl(\mathbf{W}_s^{(\ell)} \otimes_c \exp_{\mathbf{0}}^{D,c}\!\bigl(\mathbf{w}_{0,i}^{(\ell)}\bigr)\Bigr)\right]_{i=1}^{d_a},
  \label{eq:hga}
\end{equation}
where $\mathbf{W}_s^{(\ell)} = \operatorname{diag}(\mathbf{s}^{(\ell)})$ is a learnable diagonal scaling matrix, introducing $2d_a$ parameters per layer for the Q/V projections.

\begin{proposition}[Row-wise radius decomposition of diagonal M\"{o}bius scaling]
\label{prop:radius_decomp}
Let $\mathbf{W}_s^{(\ell)} = \operatorname{diag}(\mathbf{s}^{(\ell)})$ with $\mathbf{s}^{(\ell)} \in \mathbb{R}_{>0}^{d_a}$, and let $\mathbf{w}_i$ denote the $i$-th row of $\mathbf{W}_0^{(\ell)}$ (equivalently, $\log_{\mathbf{0}}^{D,c}(\mathbf{w}_{D,0,i}^{(\ell)})$). Then:
\begin{equation}
  \operatorname{Rad}_{\mathbf{w}_{D,i}^{(\ell)}} = s_{\text{eff},i}^{(\ell)} \cdot \operatorname{Rad}_{\mathbf{w}_{D,0,i}^{(\ell)}}, \quad s_{\text{eff},i}^{(\ell)} := \frac{\lVert \mathbf{w}_i \odot \mathbf{s}^{(\ell)} \rVert}{\lVert \mathbf{w}_i \rVert}.
  \label{eq:prop1}
\end{equation}
\end{proposition}
Proposition~\ref{prop:radius_decomp} shows that the diagonal scaling vector $\mathbf{s}^{(\ell)}$ induces row-wise effective radius modulation in the hyperbolic embeddings of the weight rows. The modulation is not an independently specified row-wise target; rather, each effective scaling factor $s_{\text{eff},i}^{(\ell)}$ is jointly determined by the diagonal vector and the original weight distribution of that row.

A key geometric insight justifies the diagonal parameterization: granularity control requires radial scaling, not rotation. Scaling moves weight rows along the radial dimension in the Poincar\'{e} ball (changing their granularity level) while preserving their relative angular structure (preserving the semantic relationships encoded in the pre-trained weights). Rotation, by contrast, would alter what the representations attend to, an undesirable side effect. Since only scaling is needed, a diagonal matrix ($O(d_a)$ parameters) is the geometrically sufficient parameterization, not a sparse approximation of a fuller matrix.

Because $\mathbf{W}_{\text{HGA}}^{(\ell)}$ depends on $\mathbf{s}^{(\ell)}$ through row-wise hyperbolic modulation (Proposition~\ref{prop:radius_decomp}), varying $\mathbf{s}^{(\ell)}$ changes the Euclidean Q/V projections seen by each encoder layer, inducing a layer-dependent bias on token representations toward different granularity regimes.

In implementation, the only trainable parameters are diagonal vectors $\mathbf{s}^{(\ell)} \in \mathbb{R}_{>0}^{d_a}$ (parameterized via $\operatorname{softplus}$ to ensure positivity). During each forward pass, $\mathbf{W}_{\text{HGA}}^{(\ell)}$ is computed dynamically from $\mathbf{s}^{(\ell)}$ and frozen $\mathbf{W}_0^{(\ell)}$ via Equation~(\ref{eq:hga}), then applied to input $\mathbf{X}^{(\ell)}$ via standard Euclidean matrix multiplication. Since $\mathbf{W}_s^{(\ell)} = \operatorname{diag}(\mathbf{s}^{(\ell)})$ is diagonal, the computation reduces to element-wise scaling of each row of $\log_{\mathbf{0}}^{D,c}(\mathbf{W}_0^{(\ell)})$ by $\mathbf{s}^{(\ell)}$, followed by mapping back to Euclidean space. This is efficiently implemented as several element-wise operations and $\tanh/\operatorname{artanh}$ calls, without additional large matrix storage overhead. Inputs and outputs of the entire forward pass are standard Euclidean tensors; hyperbolic geometry manifests in the parameterization of the weights: each layer's scaling vector $\mathbf{s}^{(\ell)}$ determines the row-wise radial displacement in the Poincar\'{e} ball, as characterized by Proposition~\ref{prop:radius_decomp}.

\subsubsection{HGA Injection into the Audio Encoder}

\method injects HGA only into the query and value projections of each self-attention layer, replacing the frozen weights $\mathbf{W}_0^{Q/V,(\ell)}$ with hyperbolic-modulated $\mathbf{W}_{\text{HGA}}^{Q/V,(\ell)}$ (Equation~(\ref{eq:hga})) while leaving the key projection $\mathbf{W}_0^{K,(\ell)}$ unchanged. The asymmetric design is deliberate: K determines where the model attends (the pre-trained attention anchor), while Q and V determine at what granularity the model queries and extracts information. Freezing K preserves the encoder's original attention patterns; modulating Q/V shifts the representational granularity of what is queried and extracted, aligning with the goal of granularity-aware adaptation without altering the encoder's attention structure.
Each encoder layer additionally carries a lightweight bottleneck residual adapter as a Euclidean residual branch, providing local nonlinear adaptation that complements HGA's hyperbolic weight-space modulation.

\subsubsection{Parameter Efficiency Analysis}

Under the diagonal parameterization, the Q/V projections of HGA at layer $\ell$ introduce $2d_a$ diagonal scaling parameters. Combined with a residual adapter of bottleneck dimension $r$ (up/down projections each of size $d_a \times r$), the parameter count per layer is $2d_a + 2d_a r = 2d_a(1+r)$. For Qwen2-Audio with 32 layers, hidden dimension 1280, and $r=12$, HGA across all layers requires approximately 1.06M trainable parameters, only 0.013\% of the full Qwen2-Audio-7B parameter count. Adding EMCA (approximately 5.35M) and the base model's native multimodal projector, the total trainable parameters of \method amount to approximately 0.12\% of Qwen2-Audio-7B. All other parameter efficiency figures reported in this paper refer to the full framework, not HGA alone.

\subsection{Emotion-aware Multi-capacity Cross-modal Aggregator (EMCA)}
\label{sec:emca}

EMCA is a task-level lightweight aggregator that compresses the utterance-level hidden states produced by HGA into audio prefixes for the frozen LLM. Unlike HGA, which adapts the encoder weights in hyperbolic space, EMCA operates as an utterance-level semantic readout module. Its role is to summarize frame representations, apply task-conditioned gating, and fuse multiple bottleneck branches before prefix construction.

\subsubsection{Overall EMCA Pipeline}

Given audio hidden states $\mathbf{A} \in \mathbb{R}^{T \times d_a}$ and the instruction prompt embeddings $\mathbf{T} \in \mathbb{R}^{N \times d_t}$ (token embeddings of the fixed emotion classification prompt), EMCA produces audio prefixes via:
\begin{equation}
  \mathbf{Z}_{\text{audio}} =
  \Pi\!\left(
  \operatorname{HypFuse}
  \bigl(
  \phi_1(\tilde{\mathbf{a}}),
  \phi_2(\tilde{\mathbf{a}}),
  \phi_3(\tilde{\mathbf{a}})
  \bigr)
  \right).
  \label{eq:emca}
\end{equation}
where $\operatorname{HypFuse}(\cdot)$ denotes the three-branch hyperbolic fusion (Section~\ref{sec:hgfuse}). The three branches share a common summary vector and extract complementary emotion representations at different bottleneck capacities.

\subsubsection{Task-Conditioned Gating}

Since \method operates on pure audio without ASR transcription, the instruction prompt serves not as text semantics but as a task-conditioned gate. Speech signals contain substantial emotion-irrelevant acoustic detail (channel noise, speaker identity artifacts, and recording conditions); element-wise gating between the pooled prompt embedding $\bar{\mathbf{t}}$ and the audio summary activates channels relevant to the queried emotion categories while suppressing task-irrelevant noise. We implement this via Hadamard multiplication, introducing only one $d_a \times d_t$ projection matrix and one $d_a$ bias vector-about one-third the parameter cost of cross-attention. EMCA first applies global average pooling to audio and prompt sequences to obtain summary vectors, then produces the task-conditioned audio summary $\tilde{\mathbf{a}}$:
\begin{align}
  \bar{\mathbf{a}} &= \frac{1}{T}\sum_{t=1}^{T}\mathbf{A}_t,
  \qquad
  \bar{\mathbf{t}} = \frac{1}{N}\sum_{n=1}^{N}\mathbf{T}_n, \notag\\
  \tilde{\mathbf{a}}
  &=
  \bar{\mathbf{a}} \odot (\mathbf{W}_{\text{text}}\bar{\mathbf{t}} + \mathbf{b}_{\text{text}}).
  \label{eq:modulation}
\end{align}

Temporal structure is modeled primarily within the audio encoder, where HGA modulates weights across layers and self-attention captures frame-level dependencies throughout the $L$ layers; EMCA complements this by aggregating the resulting frame representations into a compact utterance-level summary, with semantic abstraction diversity provided by its three bottleneck branches (Section~\ref{sec:hgfuse}).

\subsubsection{Multi-Capacity Hyperbolic Integration}
\label{sec:hgfuse}

Three MLP branches with bottleneck ratios $1/8$, $1/16$, and $1/32$ extract features at decreasing capacity levels from $\tilde{\mathbf{a}}$ in parallel, then map them to the Poincar\'{e} ball and fuse them via the Einstein midpoint operator, which better preserves radial distinctions among branches than Euclidean weighted averaging. The branch weights $\boldsymbol{\alpha} = \operatorname{softmax}(\mathbf{W}_{\text{route}}\tilde{\mathbf{a}}) \in \Delta^2$ are input-adaptively computed from the task-conditioned summary $\tilde{\mathbf{a}}$ via a learnable routing matrix $\mathbf{W}_{\text{route}} \in \mathbb{R}^{3 \times d_a}$, endowing the fusion with dynamic branch selection rather than fixed importance weights:
\begin{align}
  \mathbf{h}^{(i)} &= \phi_i(\tilde{\mathbf{a}}),\ i \in \{1,2,3\},\quad \boldsymbol{\alpha} = \operatorname{softmax}(\mathbf{W}_{\text{route}}\tilde{\mathbf{a}}), \notag\\
  \mathbf{p}_{\text{fuse}}
  &= \operatorname{EinMid}\!\bigl(
  \exp_{\mathbf{0}}^{D,c}(\mathbf{h}^{(1)}),
  \exp_{\mathbf{0}}^{D,c}(\mathbf{h}^{(2)}),
  \exp_{\mathbf{0}}^{D,c}(\mathbf{h}^{(3)});\,
  \boldsymbol{\alpha}, c
  \bigr) \in \mathcal{D}_c^{d_h}, \notag\\
  \mathbf{u}
  &= \log_{\mathbf{0}}^{D,c}(\mathbf{p}_{\text{fuse}}).
  \label{eq:hgfuse}
\end{align}
$\operatorname{EinMid}$ is computed via the standard Klein-to-Poincar\'{e} back-projection~\cite{ganea2018hnn}.

Branches capturing more discriminative emotion features lie closer to the Poincar\'{e} ball boundary; the Einstein midpoint assigns higher implicit weight to these boundary-proximal points via the Lorentz factor $\gamma_i = (1 - c\lVert\mathbf{k}_i\rVert^2)^{-1/2}$ (where $\mathbf{k}_i$ is the Klein image of $\mathbf{p}_i$): as a point approaches the boundary, its Lorentz factor grows, increasing its contribution to the midpoint. This geometry-native weighting mechanism automatically prioritizes fine-grained emotional semantics during fusion, a property absent in flat Euclidean averaging, which weights branches purely by the learned scalar $\boldsymbol{\alpha}$ without geometric amplification. Ablation rows G1 and G3 in Table~\ref{tab:ablation} confirm that replacing hyperbolic fusion with Euclidean averaging degrades performance by up to $4.08$~pp WA.

\subsubsection{Audio Prefix Construction and Parameter Efficiency}

The fusion vector $\mathbf{u}$ is back-projected and expanded by prefix generator $\Pi$ into $K$ audio prefix tokens, then projected by the multimodal projector $\operatorname{MMP}$ to the LLM input dimension:
\begin{align}
  \mathbf{u}_{\text{audio}} &= \mathbf{W}_{\text{proj}}\mathbf{u} + \mathbf{b}_{\text{proj}}, \notag\\
  \mathbf{Z}_{\text{audio}} &= \Pi(\mathbf{u}_{\text{audio}}), \notag\\
  \mathbf{H}_{\text{audio}} &= \operatorname{MMP}(\mathbf{Z}_{\text{audio}}; \theta_{\text{proj}}).
  \label{eq:prefix}
\end{align}
The prefix generator is defined as $\Pi(\mathbf{u}_{\text{audio}}) = \mathbf{Q}_{\text{pfx}} + \mathbf{1}_K \mathbf{u}_{\text{audio}}^{\top} \in \mathbb{R}^{K \times d_a}$, where $\mathbf{Q}_{\text{pfx}} \in \mathbb{R}^{K \times d_a}$ is a learnable prefix query matrix, $\mathbf{1}_K \in \mathbb{R}^K$ is the all-ones vector, and $\mathbf{u}_{\text{audio}} \in \mathbb{R}^{d_a}$ is the back-projected fusion vector.

\subsection{Training Objective}
\label{sec:training}

\method jointly optimizes $\theta_{\text{HGA}}$, $\theta_{\text{EMCA}}$, and $\theta_{\text{proj}}$ via a composite loss:
\begin{equation}
  \mathcal{L} = \mathcal{L}_{\text{CE}} + \lambda_{\text{hyp}}\,\mathcal{L}_{\text{hyp}} + \lambda_{\text{radius}}\,\mathcal{L}_{\text{radius}},
  \label{eq:total_loss}
\end{equation}
where $\mathcal{L}_{\text{CE}}$ is the standard label cross-entropy loss. $\mathcal{L}_{\text{hyp}}$ applies cross-entropy on geodesic distances from the EMCA fusion point $\mathbf{p}_{\text{fuse}} \in \mathcal{D}_c^{d_h}$ (Section~\ref{sec:hgfuse}) to learnable class prototypes. Each prototype $\mathbf{v}_k \in \mathbb{R}^{d_h}$ is parameterized in the tangent space at the origin and projected onto the Poincar\'{e} ball during the forward pass, ensuring numerical stability under standard AdamW optimization:
\begin{equation}
  \mathcal{L}_{\text{hyp}} = \operatorname{CE}\!\left(\left\{-\tfrac{1}{\tau}\,d_c\!\left(\mathbf{p}_{\text{fuse}},\,\exp_{\mathbf{0}}^{D,c}(\mathbf{v}_k)\right)\right\}_{k},\, y\right).
  \label{eq:loss_hyp}
\end{equation}
Since $\exp_{\mathbf{0}}^{D,c}$ maps any finite $\mathbf{v}_k$ strictly inside the ball, and $\mathbf{p}_{\text{fuse}}$ (as an Einstein midpoint of exp-mapped branch outputs) also lies strictly inside, pairwise geodesic distances remain bounded and temperature-scaled logits ($\tau = 0.07$) stay numerically stable.

$\mathcal{L}_{\text{radius}}$ enforces a monotonic radius ordering across the three EMCA branches (high $\to$ low capacity) via a hyperbolic margin ranking formulation:
\begin{equation}
  \mathcal{L}_{\text{radius}} = \sum_{i=1}^{2}\max\!\bigl(0,\; m - (\bar{r}_i - \bar{r}_{i+1})\bigr), \quad \bar{r}_i = \frac{1}{B}\sum_{b=1}^{B} \operatorname{Rad}\!\bigl(\exp_{\mathbf{0}}^{D,c}(\mathbf{h}_b^{(i)})\bigr),
  \label{eq:loss_radius}
\end{equation}
where $\mathbf{h}_b^{(i)} = \phi_i(\tilde{\mathbf{a}}_b) \in \mathbb{R}^{d_h}$ is the tangent-space MLP output of branch $i$ for sample $b$.

$\mathcal{L}_{\text{radius}}$ enforces structural ordering among the three EMCA branches, ensuring that higher-capacity branches occupy larger radii than lower-capacity ones. $\mathcal{L}_{\text{hyp}}$ anchors this ordering to emotion semantics by pulling $\mathbf{p}_{\text{fuse}}$ toward the correct class prototype. Combining both losses ensures that radial structure is aligned with emotion discriminability; either loss alone is insufficient, as confirmed by ablation rows L1 and L2 in Table~\ref{tab:ablation}.

\section{Experiments}

\subsection{Experimental Setup}
\label{sec:setup}

\paragraph{Datasets.}
Main experiments are conducted on two standard SER benchmarks: IEMOCAP~\cite{iemocap2008} (4-class, LOSO cross-validation) and MELD~\cite{meld2019} (7-class, standard train/dev/test split). The cross-dataset generalization experiment (Section~\ref{sec:cross_dataset}) transfers MELD-trained models to IEMOCAP, RAVDESS~\cite{ravdess2018}, and SAVEE~\cite{savee2011} in a zero-shot manner, and compares \method against the Qwen2-Audio zero-shot baseline under the same WA metric.

\paragraph{Baselines.}
We compare against baselines at three levels, following~\cite{li2025emorl}. This single-backbone, MELD+IEMOCAP evaluation protocol is the prevailing paradigm in LALM-based SER, adopted by recent works including EMO-RL~\cite{li2025emorl} and EmoAgent-R1~\cite{fang2026emoagent}, ensuring cross-study comparability. Non-LALM methods include HuBERT large~\cite{hsu2021}, WavLM large~\cite{chen2022wavlm}, data2vec 2.0 large~\cite{baevski2023efficient}, Whisper large V3~\cite{radford2023robust}, and Emotion2vec+ large~\cite{ma2023emotion2vec}, all fine-tuned with classification heads; results are taken from EmoBox~\cite{ma2024emobox}. LALM zero-shot baselines use Qwen2-Audio with direct inference and chain-of-thought prompting. Euclidean PEFT baselines use LoRA~\cite{lora2022} and Adapter~\cite{adapters2019} on the same Qwen2-Audio backbone with identical training configurations as \method, each with approximately 10.49M trainable parameters. The full-parameter SFT+IR reference from EMO-RL~\cite{li2025emorl} is listed in Table~\ref{tab:main} for scale only, given its different training regime.

\paragraph{Implementation Details.}
The backbone is Qwen2-Audio-7B-Instruct. All modules are trained in BFloat16 precision on 4 NVIDIA RTX 4090 24G GPUs. The optimizer is AdamW with learning rate $1.125\times10^{-3}$ and batch size 32 (via gradient accumulation). All experiments are trained for 30 epochs with an early stopping strategy of patience = 5, using validation WA to select the best checkpoint. Default hyperparameters: curvature $c = 0.01$, HGA bottleneck dimension $r = 12$, EMCA hidden dimension $d_h = 256$, number of audio prefix tokens $K = 4$, auxiliary loss weights $\lambda_{\text{hyp}} = 0.1$ and $\lambda_{\text{radius}} = 0.01$, radius ranking margin $m = 0.5$, hyperbolic prototype loss temperature $\tau = 0.07$. The emotion classification instruction prompt $\mathbf{x}_t$ follows a fixed template listing the target emotion categories for each dataset (e.g., ``\textit{Identify the emotion from: neutral, surprise, fear, sadness, joy, disgust, anger}'' for MELD); this prompt is shared across all training samples of a given dataset.

\paragraph{Evaluation Metrics.}
Following domain conventions~\cite{li2025emorl}, three metrics are reported: Weighted Accuracy (WA), which measures overall classification accuracy; Unweighted Accuracy (UA), which computes the equal-weight average over all classes and better reflects minority class performance on imbalanced datasets; and macro-averaged F1 (F1), which provides a comprehensive evaluation insensitive to class distribution.

\subsection{Main Results}

\begin{table*}[t]
  \centering
  \fontsize{8.2pt}{10.8pt}\selectfont
  \setlength{\tabcolsep}{6.5pt}
  \caption{\textbf{Comparison on MELD (7-class) and IEMOCAP (4-class) (\%).} Non-LALM results from EmoBox~\cite{ma2024emobox}. LALM+PEFT methods share identical training configurations. $\dagger$: full-parameter reference from~\cite{li2025emorl}. \textbf{Bold}: best PEFT.}
  \label{tab:main}
  \begin{tabular}{llccccccc}
    \toprule
    \multirow{2}{*}{\textbf{Model Type}} & \multirow{2}{*}{\textbf{Method}} & \multicolumn{3}{c}{\textbf{MELD}} & \multicolumn{3}{c}{\textbf{IEMOCAP}} \\
    \cmidrule(lr){3-5} \cmidrule(lr){6-8}
    & & UA & WA & F1 & UA & WA & F1 \\
    \midrule
    \rowcolor{rowgray}
    \multicolumn{8}{c}{\textbf{\textit{Non-LALM Methods}}} \\
    Non-LALM & HuBERT large            & 24.13 & 46.37 & 24.99 & 67.42 & 66.69 & 67.24 \\
    Non-LALM & WavLM large             & 28.18 & 49.31 & 29.11 & 69.47 & 69.07 & 69.29 \\
    Non-LALM & data2vec 2.0 large      & 26.33 & 47.72 & 27.35 & 57.30 & 56.23 & 56.70 \\
    Non-LALM & Whisper large V3        & 31.54 & 51.89 & 32.95 & 73.54 & 72.86 & 73.11 \\
    Non-LALM & Emotion2vec+ large      & 28.03 & 51.88 & --    & 70.70 & 67.30 & --    \\
    \midrule
    \rowcolor{rowgray}
    \multicolumn{8}{c}{\textbf{\textit{LALM Methods}}} \\
    LALM     & Qwen2-Audio (Direct Inference) & 18.96 & 39.83 & 19.84 & 53.76 & 51.52 & 47.68 \\
    LALM     & Qwen2-Audio (CoT Inference)    & 26.89 & 50.57 & 28.05 & 64.33 & 60.37 & 61.61 \\
    LALM-FT$\dagger$ & SFT + IR~\cite{li2025emorl} & 33.26 & 57.39 & 35.77 & 85.70 & 83.87 & 84.53 \\
    \midrule
    \rowcolor{rowgray}
    \multicolumn{8}{c}{\textbf{\textit{LALM + PEFT Methods}}} \\
    LALM+PEFT & Adapter                & 45.07 & 62.09 & 47.72 & 71.72 & 74.59 & 72.85 \\
    LALM+PEFT & LoRA                   & 46.88 & 65.43 & 47.84 & 78.90 & \best{80.60} & \best{79.28} \\
    \rowcolor{rowhighlight}
    LALM+PEFT & \method (ours)         & \best{50.59} & \best{68.97} & \best{53.32} & \best{82.13} & 79.08 & 78.38 \\
    \bottomrule
  \end{tabular}
\end{table*}

Under a 0.12\% parameter budget, \method outperforms both Euclidean PEFT baselines on MELD across all three metrics, with the largest gain on F1 (+5.48~pp over LoRA). On IEMOCAP, an instructive trade-off emerges: WA is marginally below LoRA ($-$1.52~pp) while UA rises by +3.23~pp, indicating that hyperbolic geometry redistributes representational capacity away from majority classes toward minority ones. This is consistent with the geometric prior: hyperbolic space allocates exponentially more volume at larger radii, naturally counteracting the feature compression that Euclidean PEFT imposes on underrepresented classes.

\subsection{Ablation Study}
\label{sec:ablation}

All ablation experiments report WA/UA/F1 on the MELD test set. Table~\ref{tab:ablation} decomposes component contributions along the module, loss, and geometry dimensions.

\begin{table}[t]
  \centering
  \fontsize{7.8pt}{9.2pt}\selectfont
  \setlength{\tabcolsep}{3.5pt}
  \caption{\textbf{Ablation study on MELD (\%).} \textit{Top}: M1 and G2 isolate hyperbolic vs.\ Euclidean HGA under identical conditions (no EMCA, $\mathcal{L}_{\text{CE}}$ only). \textit{Middle}: loss ablations on the full HGA+EMCA architecture. \textit{Bottom}: architecture-matched Euclidean counterparts ($\mathcal{L}_{\text{CE}}$ only; see text for why auxiliary losses are inapplicable). \best{Bold}: best.}
  \label{tab:ablation}
  \begin{tabular}{lccc}
    \toprule
    \textbf{Configuration} & \textbf{WA} & \textbf{UA} & \textbf{F1} \\
    \midrule
    \rowcolor{rowgray}
    \multicolumn{4}{c}{\textbf{\textit{Geometry Effect (no EMCA, $\mathcal{L}_{\text{CE}}$ only)}}} \\
    M1: Hyperbolic HGA only          & 66.16 & 49.77 & 53.36 \\
    G2: Euclidean HGA only           & 64.07 & 50.96 & 51.35 \\
    \midrule
    \rowcolor{rowgray}
    \multicolumn{4}{c}{\textbf{\textit{Geometric Supervision Effect (HGA + EMCA)}}} \\
    L0: $\mathcal{L}_{\text{CE}}$ only           & 63.90 & 45.95 & 49.47 \\
    L1: $+\,\mathcal{L}_{\text{hyp}}$  & 64.17 & 46.50 & 49.70 \\
    L2: $+\,\mathcal{L}_{\text{radius}}$ & 65.61 & 48.87 & 50.85 \\
    \rowcolor{rowhighlight}
    \method (full)                  & \best{68.97} & 50.59 & \best{53.32} \\
    \midrule
    \rowcolor{rowgray}
    \multicolumn{4}{c}{\textbf{\textit{Euclidean Counterparts ($\mathcal{L}_{\text{CE}}$ only)}}} \\
    G1: Euclidean EMCA only (no HGA)  & 64.89 & 46.62 & 50.08 \\
    G3: Euclidean HGA + Euclidean EMCA & 65.34 & 47.32 & 50.22 \\
    \bottomrule
  \end{tabular}
\end{table}

\paragraph{Geometry effect.}
Comparing M1 (Hyperbolic HGA only) with G2 (Euclidean HGA only), both using $\mathcal{L}_{\text{CE}}$ and no EMCA, isolates the effect of the parameter space. Hyperbolic weight-space modulation yields +2.09~pp WA (66.16 vs.\ 64.07), indicating a stronger inductive bias for audio encoder adaptation than its Euclidean counterpart.

\paragraph{Geometric supervision effect.}
Adding EMCA to HGA without geometric losses (L0, 63.90) drops below HGA-only performance (M1, 66.16). Without radius ordering supervision, the three parallel branches produce redundant representations that interfere with HGA's outputs. Among individual loss terms, $\mathcal{L}_{\text{radius}}$ (L2, 65.61) contributes the most by establishing the branch-separation prior; $\mathcal{L}_{\text{hyp}}$ (L1, 64.17) alone provides marginal gain but combines with $\mathcal{L}_{\text{radius}}$ to anchor the ordered structure to emotion semantics. The full combination brings \method to 68.97~WA (+5.07~pp over L0), indicating that multi-granularity aggregation in hyperbolic space requires geometry-aware supervision to be effective.

\paragraph{Euclidean counterparts.}
G1 and G3 replace all hyperbolic operations with Euclidean ones under identical architecture and parameter count, trained with $\mathcal{L}_{\text{CE}}$ only. The auxiliary losses $\mathcal{L}_{\text{hyp}}$ and $\mathcal{L}_{\text{radius}}$ are not applied because they are defined on Poincar\'{e} ball geometry and have no meaningful Euclidean analogue. G3 (Euclidean HGA + EMCA, 65.34) falls $-$3.63~pp below full \method (68.97), indicating that the performance gap originates from hyperbolic geometry and its supervision, not from structural complexity alone.

\subsection{Cross-Dataset Generalization}
\label{sec:cross_dataset}

To examine the zero-shot transfer capability of \method's hyperbolic representations, we transfer models trained on MELD directly to RAVDESS, SAVEE, and IEMOCAP without any target domain adaptation. Following~\cite{li2025emorl}, the instruction prompt at test time lists the target dataset's emotion categories while all model weights remain fixed from MELD training. We compare \method against non-LALM baselines, the full-parameter SFT+IR reference from EMO-RL~\cite{li2025emorl}, and the Qwen2-Audio zero-shot baseline. These target datasets differ from MELD across multiple dimensions, including recording environment, speaker scale, emotion taxonomy, and language register, making the setting a rigorous test of representational generalization under cross-corpus shift.

\begin{table}[t]
  \centering
  \fontsize{7.8pt}{9.2pt}\selectfont
  \setlength{\tabcolsep}{3pt}
  \caption{\textbf{Zero-shot cross-dataset generalization (trained on MELD, WA \%).} No target domain fine-tuning. $\dagger$: full-parameter reference from~\cite{li2025emorl}. \textbf{Bold}: best excluding full-parameter reference.}
  \label{tab:cross}
  \begin{tabular}{llccc}
    \toprule
    \textbf{Model Type} & \textbf{Method} & \makecell{\textbf{RAVDESS}\\\textbf{WA}} & \makecell{\textbf{SAVEE}\\\textbf{WA}} & \makecell{\textbf{IEMOCAP}\\\textbf{WA}} \\
    \midrule
    \rowcolor{rowgray}
    \multicolumn{5}{c}{\textbf{\textit{Non-LALM Methods}}} \\
    Non-LALM & HuBERT large       & 25.02 & 31.54 & 44.60 \\
    Non-LALM & WavLM large        & 33.90 & 34.10 & 48.59 \\
    Non-LALM & data2vec 2.0 large & 34.21 & 37.79 & 47.43 \\
    Non-LALM & Whisper large V3   & 40.68 & 42.18 & 46.14 \\
    \midrule
    LALM-FT$\dagger$ & SFT + IR~\cite{li2025emorl} & 59.83 & 71.52 & 82.74 \\
    \midrule
    LALM & Qwen2-Audio (Zero-shot) & 50.42 & 30.83 & 51.52 \\
    \rowcolor{rowhighlight}
    LALM+PEFT & \method\ (ours) & \best{56.33} & \best{66.67} & \best{76.14} \\
    \bottomrule
  \end{tabular}
\end{table}

As shown in Table~\ref{tab:cross}, \method outperforms all non-LALM baselines and the Qwen2-Audio zero-shot baseline across all three out-of-domain corpora without target-domain adaptation. The largest gain appears on SAVEE: the zero-shot baseline reaches 30.83 WA, whereas \method achieves 66.67 WA (+35.8~pp), and exceeds the best non-LALM baseline (Whisper, 42.18) by +24.5~pp. On RAVDESS and IEMOCAP, \method also improves over the zero-shot baseline by +5.9~pp and +24.6~pp, respectively. We also observe that the zero-shot baseline varies substantially across corpora (50.42 on RAVDESS vs.\ 30.83 on SAVEE), while \method remains comparatively stable, suggesting that hyperbolic representations may encode emotion structure that transfers better across recording conditions.

\subsection{Per-Class Analysis}

\method improves F1 on six of seven MELD classes (Table~\ref{tab:per_class}). The most dramatic gains appear on neutral (+75.9~pp) and joy (+43.1~pp), classes where the zero-shot baseline exhibits severe under-prediction due to class imbalance. The sole exception is surprise, whose F1 drops from 77.20\% to 57.90\%. This decline reflects rebalancing rather than model failure: the zero-shot baseline inflates surprise F1 by over-predicting it at the expense of neutral and joy. After adaptation, predictions distribute more evenly, with surprise's main confusions shifting to neutral (15.7\%) and joy (12.5\%); the model still retains 61.31\% recall on surprise, confirming the class remains well-represented in the learned manifold.

\begin{table}[t]
  \centering
  \fontsize{7.5pt}{8.8pt}\selectfont
  \setlength{\tabcolsep}{2.8pt}
  \caption{\textbf{Per-class F1 on MELD (\%).} Comparison between zero-shot baseline and \method.}
  \label{tab:per_class}
  \begin{tabular}{lccccccc}
    \toprule
    \textbf{Method} & \textbf{ang.} & \textbf{dis.} & \textbf{fear} & \textbf{joy} & \textbf{neu.} & \textbf{sad.} & \textbf{sur.} \\
    \midrule
    Qwen2-Audio Base & 42.00 & 23.50 & 30.00 & 22.10 & 6.60  & 28.40 & 77.20 \\
    \rowcolor{rowhighlight}
    \method (ours)   & \best{55.16} & \best{33.30} & \best{36.05} & \best{65.18} & \best{82.50} & \best{43.14} & 57.90 \\
    \bottomrule
  \end{tabular}
\end{table}

\subsection{Geometric Analysis of Hyperbolic Space}

\subsubsection{Verification of Hyperbolic Structure in Audio Encoder Representations}
\label{sec:geo_analysis}

Audio encoder representations evolve from coarse-grained acoustic patterns (smaller hyperbolic radius) to fine-grained emotion semantics (larger radius) in the Poincar\'{e} ball, motivating HGA's radius modulation. We compute the Gromov $\delta$-hyperbolicity index (smaller $\delta$ indicates geometry closer to tree-like hyperbolic structure) of representations from each layer of the Qwen2-Audio audio encoder and compare with the $\delta$ values of random Gaussian vectors of the same dimension. At all sampled layers, $\delta_\text{audio}$ is substantially smaller than the random baseline $\delta_\text{random}$ (full per-layer plot in Appendix~E), The ratio $\delta_\text{audio}/\delta_\text{random}$ reaches as low as 0.014 at layer 15, meaning audio representations are approximately 70 times closer to hyperbolic geometry than random vectors; even at the layer with the highest ratio (layer 27), the ratio is only 0.273, indicating that hyperbolic structure pervades the entire encoder and providing direct empirical evidence for \method's choice of $\mathcal{D}_c^{d_a}$ as the adaptation space.

\subsubsection{Layer-Adaptive Hyperbolic Radius Modulation by HGA}

Figure~\ref{fig:hga_radius} shows the distribution of diagonal scaling parameters $\mathbf{s}^{(\ell)}$ (Equation~(\ref{eq:hga})) learned by HGA after \method training. Scaling factors vary substantially across layers rather than converging to a uniform constant. Shallow layers ($\ell < 10$) tend toward smaller scaling factors, while deep layers ($\ell > 20$) tend toward larger ones, consistent with the geometric prediction of Proposition~\ref{prop:radius_decomp}. This gradient confirms that HGA learns layer-adaptive radius modulation aligned with the encoder's representational hierarchy.

\begin{figure}[t]
  \centering
  \includegraphics[width=\columnwidth]{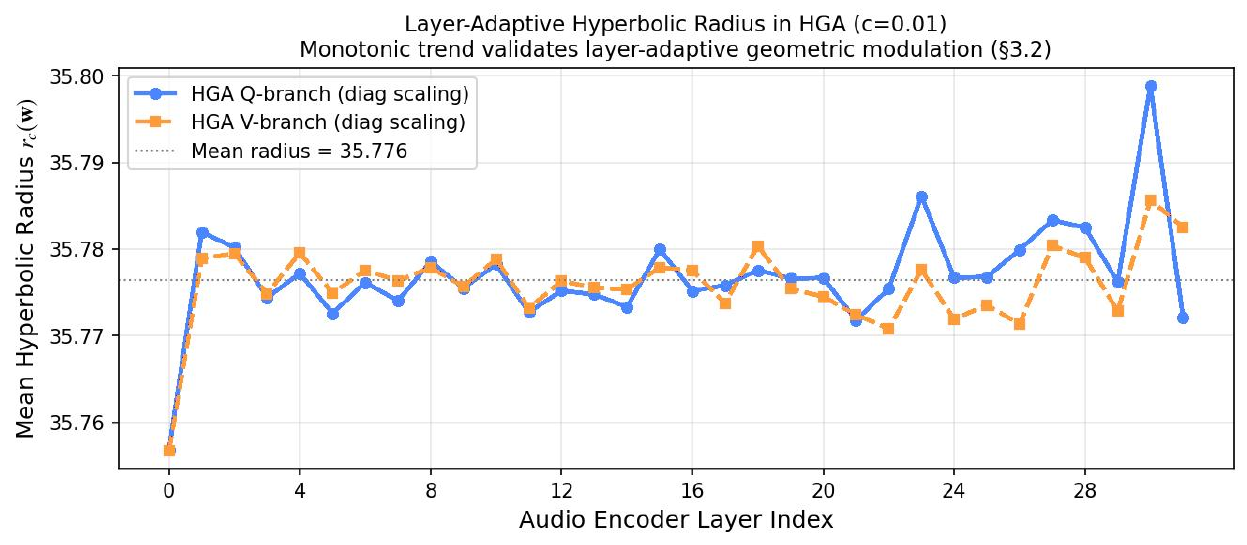}
  \caption{Distribution of HGA diagonal scaling parameters $\mathbf{s}^{(\ell)}$ across layers on MELD. Shallow layers exhibit smaller scaling (contraction, coarse-grained) while deep layers exhibit larger scaling (expansion, fine-grained), consistent with hierarchical geometry.}
  \label{fig:hga_radius}
\end{figure}

\section{Limitations and Broader Impact}

While HyPASE demonstrates robust generalization, future work should validate transferability beyond the Qwen2-Audio backbone and English-centric benchmarks. Although tuning only 0.12\% of parameters aligns with Green AI principles, real-world deployment requires fairness auditing to ensure hyperbolic adaptation does not exacerbate demographic or accent-based disparities inherited from pre-trained encoders.

\section{Conclusion}

In this paper, we introduced HyPASE, a hyperbolic parameter-efficient fine-tuning framework that moves beyond flat Euclidean spaces for Large Audio-Language Model adaptation. By grounding layer-adaptive weight modulation (HGA) and multi-scale semantic aggregation (EMCA) within the Poincar\'{e} ball, HyPASE explicitly aligns the adaptation process with the inherent multi-granularity of speech emotion. Empirical results demonstrate that this geometric prior achieves competitive recognition performance---particularly in class-imbalanced emotion recognition---and robust zero-shot cross-corpus generalization with negligible parameter overhead, outperforming Euclidean baselines on the majority of evaluation dimensions. Ultimately, we hope this work inspires further exploration of non-Euclidean geometry in multimodal foundation models. Code, pretrained weights, and training logs are publicly available at \url{https://github.com/LilSicko/HyPase} to ensure full reproducibility.

\begin{acks}
This work was supported by the
\grantsponsor{nsfc}
{National Natural Science Foundation of China}
{}
under Grant No.~\grantnum{nsfc}{62201389},
and by the
\grantsponsor{frfcu}
{Fundamental Research Funds for the Central Universities}
{}
under Grant No.~\grantnum{frfcu}{22120230311}.
\end{acks}

\appendix
\providecommand{\Rad}{}\renewcommand{\Rad}{\operatorname{Rad}}
\providecommand{\artanh}{}\renewcommand{\artanh}{\operatorname{artanh}}
\providecommand{\warn}{}\renewcommand{\warn}[1]{\textcolor{myred}{#1}}

\section{Proofs of Theoretical Results}

We prove Theorem~1 and Proposition~1 from the main paper. Both proofs
rely on the following unifying lemma, which identifies the hyperbolic
radius of an exp-mapped tangent vector with the Euclidean norm of that
vector.

\begin{lemma}[Tangent--Radius Identity]
\label{lem:tangent_radius}
For every $\mathbf{v}\in T_{\mathbf{0}}\mathcal{D}_c^d\cong\mathbb{R}^d$,
\begin{equation}
  \Rad_{\exp_{\mathbf{0}}^{D,c}(\mathbf{v})} = \|\mathbf{v}\|.
  \label{eq:lemma}
\end{equation}
\end{lemma}

\begin{proof}
The exponential map at the origin has the closed form
\begin{equation}
  \exp_{\mathbf{0}}^{D,c}(\mathbf{v})
  = \tanh\!\Bigl(\tfrac{\sqrt{c}}{2}\|\mathbf{v}\|\Bigr)
    \frac{\mathbf{v}}{\sqrt{c}\|\mathbf{v}\|},
  \label{eq:expmap}
\end{equation}
so its Euclidean norm is
\begin{equation}
  \bigl\|\exp_{\mathbf{0}}^{D,c}(\mathbf{v})\bigr\|
  = \frac{1}{\sqrt{c}}\tanh\!\Bigl(\tfrac{\sqrt{c}}{2}\|\mathbf{v}\|\Bigr).
  \label{eq:expmap_norm}
\end{equation}
Substituting into the hyperbolic radius definition
$\Rad_{\mathbf{x}} = \tfrac{2}{\sqrt{c}}\artanh(\sqrt{c}\|\mathbf{x}\|)$
and setting $t := \tfrac{\sqrt{c}}{2}\|\mathbf{v}\|$:
\begin{align}
  \Rad_{\exp_{\mathbf{0}}^{D,c}(\mathbf{v})}
  &= \frac{2}{\sqrt{c}}\artanh\!\Bigl(\sqrt{c}
      \cdot\frac{\tanh(t)}{\sqrt{c}}\Bigr)
   = \frac{2}{\sqrt{c}}\artanh\!\bigl(\tanh(t)\bigr) \notag \\
  &= \frac{2}{\sqrt{c}}\cdot t
   = \frac{2}{\sqrt{c}}\cdot\frac{\sqrt{c}}{2}\|\mathbf{v}\|
   = \|\mathbf{v}\|,
  \label{eq:lemma_proof}
\end{align}
where the third equality uses $\artanh(\tanh(t))=t$ for all $t\in\mathbb{R}$. \qed
\end{proof}

The logarithmic map at the origin is the exact inverse:
$\log_{\mathbf{0}}^{D,c}(\mathbf{x})
  = \tfrac{2}{\sqrt{c}}\artanh(\sqrt{c}\|\mathbf{x}\|)\tfrac{\mathbf{x}}{\|\mathbf{x}\|}$,
so
\begin{equation}
  \|\log_{\mathbf{0}}^{D,c}(\mathbf{x})\|
  = \tfrac{2}{\sqrt{c}}\artanh(\sqrt{c}\|\mathbf{x}\|)
  = \Rad_{\mathbf{x}}.
  \label{eq:log_rad}
\end{equation}

\begin{theorem}[Restatement of Theorem~\ref{thm:radius_scaling}]
\label{thm:radius_scaling_app}
For any $\mathbf{x}\in\mathcal{D}_c^d$ and scalar $s>0$,
$s\otimes_c\mathbf{x} := \exp_{\mathbf{0}}^{D,c}(s\cdot\log_{\mathbf{0}}^{D,c}(\mathbf{x}))$
satisfies $\Rad_{s\otimes_c\mathbf{x}} = s\cdot\Rad_{\mathbf{x}}$.
\end{theorem}

\begin{proof}
Let $\mathbf{v} = \log_{\mathbf{0}}^{D,c}(\mathbf{x})$.
By Eq.~(\ref{eq:log_rad}), $\|\mathbf{v}\|=\Rad_{\mathbf{x}}$.
The scaled tangent vector $s\mathbf{v}$ maps back via the exponential map; its
Euclidean norm follows from Eq.~(\ref{eq:expmap_norm}) with $\|\mathbf{v}\|$ replaced by $s\|\mathbf{v}\|$:
\begin{equation}
  \bigl\|s\otimes_c\mathbf{x}\bigr\|
  = \bigl\|\exp_{\mathbf{0}}^{D,c}(s\mathbf{v})\bigr\|
  = \frac{1}{\sqrt{c}}\tanh\!\Bigl(\tfrac{\sqrt{c}}{2}\cdot s\|\mathbf{v}\|\Bigr).
  \label{eq:thm1_norm}
\end{equation}
Applying Lemma~\ref{lem:tangent_radius} to $s\mathbf{v}$:
\begin{equation}
  \Rad_{s\otimes_c\mathbf{x}}
  = \Rad_{\exp_{\mathbf{0}}^{D,c}(s\mathbf{v})}
  = \|s\mathbf{v}\|
  = s\|\mathbf{v}\|
  = s\cdot\Rad_{\mathbf{x}}. \qedhere
\end{equation}
\end{proof}

\begin{proposition}[Restatement of Proposition~\ref{prop:radius_decomp}]
\label{prop:radius_decomp_app}
Let $\mathbf{W}_s^{(\ell)}=\operatorname{diag}(\mathbf{s}^{(\ell)})$ with
$\mathbf{s}^{(\ell)}\in\mathbb{R}_{>0}^{d_a}$,
and let $\mathbf{w}_i = \log_{\mathbf{0}}^{D,c}(\mathbf{w}_{D,0,i}^{(\ell)})$
denote the tangent-space preimage of the $i$-th frozen weight row.
Then the adapted row $\mathbf{w}_{D,i}^{(\ell)}
  = \mathbf{W}_s^{(\ell)}\otimes_c\mathbf{w}_{D,0,i}^{(\ell)}$
satisfies:
\begin{equation}
  \Rad_{\mathbf{w}_{D,i}^{(\ell)}}
  = s_{\text{eff},i}^{(\ell)}\cdot\Rad_{\mathbf{w}_{D,0,i}^{(\ell)}},
  \qquad
  s_{\text{eff},i}^{(\ell)} := \frac{\|\mathbf{w}_i\odot\mathbf{s}^{(\ell)}\|}{\|\mathbf{w}_i\|}.
  \label{eq:prop1_proof}
\end{equation}
\end{proposition}

\begin{proof}
By definition of diagonal M\"{o}bius multiplication via log/exp:
\[
\mathbf{w}_{D,i}^{(\ell)}
  = \exp_{\mathbf{0}}^{D,c}\!\bigl(\mathbf{s}^{(\ell)}\odot
    \log_{\mathbf{0}}^{D,c}(\mathbf{w}_{D,0,i}^{(\ell)})\bigr)
  = \exp_{\mathbf{0}}^{D,c}(\mathbf{s}^{(\ell)}\odot\mathbf{w}_i).
\]
Applying Lemma~\ref{lem:tangent_radius} twice:
\begin{equation}
  \Rad_{\mathbf{w}_{D,i}^{(\ell)}} = \|\mathbf{s}^{(\ell)}\odot\mathbf{w}_i\|,
  \qquad
  \Rad_{\mathbf{w}_{D,0,i}^{(\ell)}} = \|\mathbf{w}_i\|.
  \label{eq:prop1_rad}
\end{equation}
Dividing yields Eq.~(\ref{eq:prop1_proof}). \qed
\end{proof}

\begin{corollary}[Effective scaling bounds]
\label{cor:bounds}
$\min_j s_j^{(\ell)} \;\le\; s_{\text{eff},i}^{(\ell)} \;\le\; \max_j s_j^{(\ell)}$.
\end{corollary}

\begin{proof}
Let $s_{\min}=\min_j s_j^{(\ell)}$ and $s_{\max}=\max_j s_j^{(\ell)}$.
Bounding the numerator:
$s_{\min}\|\mathbf{w}_i\| \le \|\mathbf{s}^{(\ell)}\odot\mathbf{w}_i\|
\le s_{\max}\|\mathbf{w}_i\|$.
Dividing by $\|\mathbf{w}_i\|$ gives the result.
\end{proof}

\begin{corollary}[HGA forward-pass simplification]
\label{cor:hga_forward}
The HGA Euclidean weight matrix satisfies
$\mathbf{W}_{\mathrm{HGA},i}^{(\ell)}
  = \log_{\mathbf{0}}^{D,c}\!\bigl(\mathbf{w}_{D,i}^{(\ell)}\bigr)
  = \mathbf{s}^{(\ell)}\odot\mathbf{w}_i$,
i.e., the full exp--M\"{o}bius--log pipeline reduces to element-wise scaling
of the frozen weight rows in tangent space.
\end{corollary}

\begin{proof}
From the proof of Proposition~\ref{prop:radius_decomp_app},
$\mathbf{w}_{D,i}^{(\ell)}
  = \exp_{\mathbf{0}}^{D,c}(\mathbf{s}^{(\ell)}\odot\mathbf{w}_i)$.
Applying $\log_{\mathbf{0}}^{D,c}$ (the exact inverse of $\exp_{\mathbf{0}}^{D,c}$)
gives $\mathbf{s}^{(\ell)}\odot\mathbf{w}_i$ directly.
\end{proof}

\section{Hyperbolic Operations Reference}

Table~\ref{tab:ops} provides a self-contained reference for all
hyperbolic operations used in \method.  Notation follows the main paper
throughout; $c>0$ is the curvature parameter.

\begin{table}[h]
  \centering
  \caption{Hyperbolic operations used in \method.}
  \label{tab:ops}
  \fontsize{7.5pt}{9.0pt}\selectfont
  \setlength{\tabcolsep}{3pt}
  \begin{tabular}{lp{5.5cm}}
    \toprule
    \textbf{Operation} & \textbf{Formula} \\
    \midrule
    Conformal factor &
      $\lambda_{c,\mathbf{x}} = \tfrac{2}{1-c\|\mathbf{x}\|^2}$ \\[2pt]
    M\"{o}bius addition &
      \makecell[l]{%
        $\mathbf{x}\oplus_c\mathbf{y} = N/D$,\\[1pt]
      $N=(1+2c\langle\mathbf{x},\mathbf{y}\rangle+c\|\mathbf{y}\|^2)\mathbf{x}$\\
        $\phantom{N=}+(1-c\|\mathbf{x}\|^2)\mathbf{y}$,\\[1pt]
        $D=1+2c\langle\mathbf{x},\mathbf{y}\rangle+c^2\|\mathbf{x}\|^2\|\mathbf{y}\|^2$%
      } \\[2pt]
    Poincar\'{e} ball &
      $\mathcal{D}_c^d = \{\mathbf{x}\in\mathbb{R}^d : c\|\mathbf{x}\|^2<1\}$ \\[2pt]
    Hyperbolic radius &
      $\Rad_{\mathbf{x}} = \tfrac{2}{\sqrt{c}}\artanh(\sqrt{c}\|\mathbf{x}\|)$ \\[2pt]
    Exp map at $\mathbf{0}$ &
      $\exp_{\mathbf{0}}^{D,c}(\mathbf{v})
        = \tanh\!\bigl(\tfrac{\sqrt{c}}{2}\|\mathbf{v}\|\bigr)
          \tfrac{\mathbf{v}}{\sqrt{c}\|\mathbf{v}\|}$ \\[2pt]
    Log map at $\mathbf{0}$ &
      $\log_{\mathbf{0}}^{D,c}(\mathbf{x})
        = \tfrac{2}{\sqrt{c}}\artanh(\sqrt{c}\|\mathbf{x}\|)
          \tfrac{\mathbf{x}}{\|\mathbf{x}\|}$ \\[2pt]
    M\"{o}bius scalar &
      $s\otimes_c\mathbf{x}
        = \exp_{\mathbf{0}}^{D,c}(s\cdot\log_{\mathbf{0}}^{D,c}(\mathbf{x}))$ \\[2pt]
    Diagonal M\"{o}bius &
      $\mathbf{D}\otimes_c\mathbf{x}
        = \exp_{\mathbf{0}}^{D,c}(\mathbf{d}\odot\log_{\mathbf{0}}^{D,c}(\mathbf{x}))$,
      $\mathbf{D}=\operatorname{diag}(\mathbf{d})$ \\[2pt]
    Geodesic distance &
      $d_c(\mathbf{x},\mathbf{y})
        = \tfrac{2}{\sqrt{c}}\artanh\!\bigl(\sqrt{c}\|\mathbf{x}\ominus_c\mathbf{y}\|\bigr)$ \\[2pt]
    Lorentz factor &
      $\gamma_i^K = (1-c\|\mathbf{k}_i\|^2)^{-1/2}$,
      $\mathbf{k}_i$: Klein image of $\mathbf{p}_i$ \\[2pt]
    Einstein midpoint &
      $\operatorname{EinMid}(\{\mathbf{p}_i\};\boldsymbol{\alpha})
        = \mathcal{K}^{-1}\!\Bigl(\tfrac{\sum_i\alpha_i\gamma_i^K\mathbf{k}_i}
                                       {\sum_i\alpha_i\gamma_i^K}\Bigr)$,
      $\mathcal{K}^{-1}$: Klein$\!\to\!$Poincar\'{e} \\[2pt]
    Klein projection &
      $\mathbf{k}_i = \tfrac{2\mathbf{p}_i}{1+c\|\mathbf{p}_i\|^2}$;
      inverse: $\mathbf{p}_i = \tfrac{\mathbf{k}_i}{1+\sqrt{1-c\|\mathbf{k}_i\|^2}}$ \\[2pt]
    Stable $\artanh$ &
      clamp input to $(-1{+}\epsilon,\, 1{-}\epsilon)$
      before $\tfrac{1}{2}\ln\tfrac{1+x}{1-x}$; use custom backward \\
    \bottomrule
  \end{tabular}
\end{table}

\noindent
All \method\ operations follow a unified \emph{lift--operate--project} pipeline:
Euclidean features are lifted to the Poincar\'{e} ball via $\exp_{\mathbf{0}}^{D,c}$,
combined via M\"{o}bius addition or diagonal scaling, and returned to tangent space
via $\log_{\mathbf{0}}^{D,c}$---a global chart that makes every operation
tractable in closed form without geodesic integration.
The Einstein midpoint used in EMCA branch fusion is formulated in the Klein model
and mapped back via $\mathcal{K}^{-1}$, exploiting its linearity over weighted sums.
The \emph{stable $\artanh$} row is critical in practice: gradient flow through the
curvature parameter $c$ amplifies near-boundary activations, and without the clamp the
backward pass diverges early in training; the custom backward uses the analytic
derivative $\tfrac{d}{dx}\artanh(x)=\tfrac{1}{1-x^2}$ evaluated at the clamped value.

\section{EMCA Parameter Count}

Table~\ref{tab:params} derives the stated $\approx$5.35M EMCA trainable
parameters under the default configuration
($d_t{=}3584,\,d_a{=}1280,\,d_h{=}256,\,K{=}4,\,C{=}7$).

\begin{table}[h]
  \centering
  \caption{EMCA parameter breakdown (default config).}
  \label{tab:params}
  \fontsize{7.5pt}{9.0pt}\selectfont
  \setlength{\tabcolsep}{4pt}
  \begin{tabular}{lcc}
    \toprule
    \textbf{Component} & \textbf{Formula} & \textbf{Params (M)} \\
    \midrule
    Task mixer $\mathbf{W}_\text{text},\mathbf{b}_\text{text}$
      & $d_t d_a + d_a$
      & $4.587+0.001=4.588$ \\
    Branch 1 ($r_1=d_a/8=160$)
      & $d_a r_1 + r_1 d_h$
      & $0.205+0.041=0.246$ \\
    Branch 2 ($r_2=d_a/16=80$)
      & $d_a r_2 + r_2 d_h$
      & $0.102+0.020=0.122$ \\
    Branch 3 ($r_3=d_a/32=40$)
      & $d_a r_3 + r_3 d_h$
      & $0.051+0.010=0.061$ \\
    Routing $\mathbf{w}_\text{scale}$
      & $3$ (learnable scalars)
      & $<0.001$ \\
    Projection $\mathbf{W}_\text{proj},\mathbf{b}_\text{proj}$
      & $d_h d_a + d_a$
      & $0.328+0.001=0.329$ \\
    Prefix query $\mathbf{Q}_\text{pfx}$
      & $K\times d_a$
      & $<0.001$ \\
    Prototypes $\{\mathbf{v}_k\}_{k=1}^C$
      & $C\times d_h$
      & $<0.001$ \\
    \midrule
    \textbf{EMCA total}
      & ---
      & $\approx\mathbf{5.35}$ \\
    \midrule
    HGA ($r{=}12$, $L{=}32$ layers, Q+V)
      & $2\times L\times 2d_a(1+r)$
      & $\approx1.06$ \\
    \textbf{Full \method\ (excl.\ projector)}
      & ---
      & $\approx\mathbf{6.41}$ \\
    \bottomrule
  \end{tabular}
\end{table}

\noindent
Adding the native multimodal projector ($\theta_\text{proj}$,
$\approx0.37$M) yields $\approx6.78$M total trainable parameters,
or $\approx0.12\%$ of Qwen2-Audio-7B (8.4B parameters).

\noindent
The text mixer $\mathbf{W}_\text{text}$ dominates the count
(4.588\,M out of 5.35\,M, or 85.8\%) because it must bridge
Qwen2-Audio's LLM text dimension ($d_t{=}3584$) down to the
audio encoder space ($d_a{=}1280$); no weight sharing with the
frozen backbone is possible here.
The three bottleneck branches intentionally use \emph{asymmetric} ranks
($d_a/8$, $d_a/16$, $d_a/32$) rather than a uniform rank so that each
branch occupies a distinct radial shell once embedded in the Poincar\'{e} ball,
providing a geometric scaffold for $\mathcal{L}_\text{radius}$ to supervise.
The prefix query $\mathbf{Q}_\text{pfx}$ ($K{\times}d_a$) and emotion
prototype vectors $\{\mathbf{v}_k\}_{k=1}^C$ ($C{\times}d_h$) together
add fewer than 10K parameters yet supply the cross-domain geometric anchors
that account for the strong zero-shot transfer results reported in the main paper.

\section{Hyperparameter Sensitivity}

All sensitivity experiments use MELD (7-class) as the test bed.
Table~\ref{tab:sens_single} sweeps four key parameters on an isolated
EMCA branch (\textit{without} HGA; $\tau{=}0.1$ code default).
Table~\ref{tab:sens_grid} shows the full
$\lambda_\text{hyp}\times\lambda_\text{radius}$ grid on the complete
\method\ (HGA enabled; $\tau{=}0.1$, baseline WA\,65.52\%,
\textit{cf}.\ main paper WA\,68.97\% which uses $\tau{=}0.07$).

\begin{table}[h]
  \centering
  \caption{Single-factor sensitivity on EMCA-only branch (HGA disabled). WA/UA/wF1 (\%).}
  \label{tab:sens_single}
  \fontsize{7.5pt}{9.0pt}\selectfont
  \setlength{\tabcolsep}{3.5pt}
  \begin{tabular}{lcccl}
    \toprule
    \textbf{Param} & \textbf{Value} & \textbf{WA} & \textbf{UA} & \textbf{Note} \\
    \midrule
    \multirow{3}{*}{Curvature $c$\;(def.\;0.01)}
      & 0.003 & 46.17 & 45.25 & flat; near-Euclidean \\
      & 0.03  & 58.16 & 46.12 & workable \\
      & \warn{0.1}   & \warn{30.38} & \warn{41.04} & \warn{collapse; too curved} \\
    \midrule
    \multirow{3}{*}{Temperature $\tau$\;(def.\;0.1)}
      & \textbf{0.05} & \textbf{65.21} & 43.16 & near-full performance \\
      & \warn{0.2}   & \warn{29.27} & \warn{41.36} & \warn{logit flatness} \\
      & \warn{0.4}   & \warn{31.15} & \warn{42.07} & \warn{logit flatness} \\
    \midrule
    \multirow{3}{*}{$\lambda_\text{hyp}$\;(def.\;0.1)}
      & 0.0   & 44.29 & 45.30 & no prototype loss \\
      & 0.05  & 33.03 & 42.92 & \textit{worse} than 0.0 \\
      & 0.2   & 30.15 & 41.53 & over-regularised \\
    \midrule
    \multirow{3}{*}{$\lambda_\text{radius}$\;(def.\;0.01)}
      & 0.0   & 49.69 & 45.99 & no branch ordering \\
      & 0.005 & 60.08 & 42.44 & recovers most gain \\
      & \warn{0.02}  & \warn{30.04} & \warn{41.55} & \warn{collapse; too rigid} \\
    \bottomrule
  \end{tabular}
\end{table}

\begin{table}[h]
  \centering
  \caption{Full \method\ $\lambda_\text{hyp}\!\times\!\lambda_\text{radius}$ grid (WA\,\%).
           Baseline (0.1, 0.01): 65.52. $\Delta$ vs.\ baseline.}
  \label{tab:sens_grid}
  \fontsize{7.5pt}{9.0pt}\selectfont
  \setlength{\tabcolsep}{3.5pt}
  \begin{tabular}{c|cccc}
    \toprule
    $\lambda_h \backslash \lambda_r$ & \textbf{0.0} & \textbf{0.005} & \textbf{0.01} & \textbf{0.02} \\
    \midrule
    \textbf{0.0}
      & 65.70\;{\scriptsize$+$0.18}
      & 65.79\;{\scriptsize$+$0.27}
      & 65.43\;{\scriptsize$-$0.09}
      & 65.52\;{\scriptsize$\pm$0} \\
    \textbf{0.05}
      & 65.34\;{\scriptsize$-$0.18}
      & \warn{61.37\;{\scriptsize$-$4.15}}
      & 65.43\;{\scriptsize$-$0.09}
      & 64.80\;{\scriptsize$-$0.72} \\
    \textbf{0.1}
      & 65.52\;{\scriptsize$\pm$0}
      & 64.98\;{\scriptsize$-$0.54}
      & \textit{65.52} (base)
      & 65.70\;{\scriptsize$+$0.18} \\
    \textbf{0.2}
      & 65.70\;{\scriptsize$+$0.18}
      & 65.43\;{\scriptsize$-$0.09}
      & \best{65.97\;{\scriptsize$+$0.45}}
      & 64.98\;{\scriptsize$-$0.54} \\
    \bottomrule
  \end{tabular}
\end{table}

\noindent\textbf{Key findings.}
(1)~Among all single-factor sweeps, $\tau$ and $\lambda_\text{radius}$
are the most dangerous: both collapse to near-random ($\sim$30\%~WA)
when set too large ($\tau\!\ge\!0.2$ or $\lambda_r\!=\!0.02$).
(2)~In the full \method\ grid, 15/16 configurations achieve WA within
$\pm$0.7~pp of the baseline, confirming that HGA provides robust
geometric regularization that cushions lambda sensitivity.
(3)~The sole outlier ($\lambda_h{=}0.05, \lambda_r{=}0.005$,
$-$4.15~pp) arises from simultaneous weakness of both auxiliary losses:
neither loss alone is strong enough to establish branch separation,
and the deficiencies compound rather than cancel.
(4)~The best grid configuration ($\lambda_h{=}0.2, \lambda_r{=}0.01$)
improves by only $+$0.45~pp; the default (0.1, 0.01) is conservative
but stable.

\noindent
The near-flat performance surface in ($\lambda_\text{hyp}$,
$\lambda_\text{radius}$) space (15 of 16 configurations within $\pm$0.7\,pp)
is consistent with the geometric interpretation: both auxiliary losses
impose \emph{mild ordering constraints}---push representations to
distinct radial shells---rather than sharp boundaries, so many weight
configurations satisfy them roughly equally.
In contrast, the temperature $\tau$ and $\lambda_\text{radius}$ single-factor
sweeps (Table~\ref{tab:sens_single}) reveal hard cliff edges at $\tau{\ge}0.2$
and $\lambda_r{=}0.02$, where logit flatness or branch collapse degrades WA
to near-chance levels ($\sim$30\%).
Practitioners should therefore fix $\tau{=}0.07$ and
$\lambda_r{\in}\{0.005, 0.01\}$ first, then optionally refine
$\lambda_\text{hyp}$ over the safe range $\{0.05, 0.1, 0.2\}$;
the joint grid search reported here confirms that any combination
in this subspace produces stable, near-optimal results.

\section{Per-Layer Gromov delta-Hyperbolicity}

Figure~\ref{fig:delta_appendix} reports the full per-layer Gromov
$\delta$-hyperbolicity comparison referenced in the main paper
(Section 4.5): the true audio representation $\delta_\text{audio}$
against the random Gaussian baseline $\delta_\text{random}$ at every
sampled layer of the Qwen2-Audio audio encoder.

\begin{figure}[h]
  \centering
  \includegraphics[width=\columnwidth]{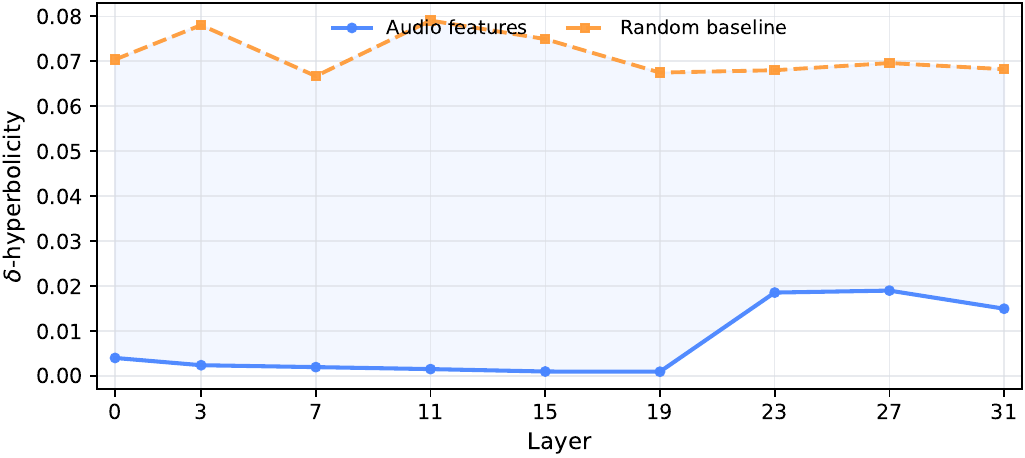}
  \caption{Gromov $\delta$-hyperbolicity of each Qwen2-Audio layer vs.\ random Gaussian baseline on MELD. Audio representations at all layers exhibit markedly lower $\delta$ than random vectors, confirming inherent hyperbolic structure.}
  \label{fig:delta_appendix}
\end{figure}

\bibliographystyle{ACM-Reference-Format}
\balance


\end{document}